\documentclass[12pt,draftcls,onecolumn]{IEEEtran}

\usepackage{amsmath,amssymb,amsthm,mathtools}
\usepackage{amsfonts,mathrsfs}
\usepackage{graphicx}
\usepackage{hyperref}
\usepackage{multirow,booktabs}
\usepackage{algorithm,algorithmicx,algpseudocode}
\usepackage{enumitem}
\usepackage{xcolor}
\usepackage{tikz}
\usetikzlibrary{positioning,arrows.meta,calc}
\usepackage{rotating}

\newtheorem{theorem}{Theorem}
\newtheorem{proposition}[theorem]{Proposition}
\newtheorem{lemma}[theorem]{Lemma}
\newtheorem{corollary}[theorem]{Corollary}
\newtheorem{definition}[theorem]{Definition}

\newtheorem{remark}[theorem]{Remark}
\newtheorem{assumption}{Assumption}

\newcommand{\KL}{D_{\mathrm{KL}}}
\newcommand{\E}{\mathbb{E}}
\newcommand{\R}{\mathbb{R}}
\newcommand{\cX}{\mathcal{X}}
\newcommand{\cY}{\mathcal{Y}}
\newcommand{\T}{\mathcal{T}}
\newcommand{\G}{\mathcal{G}}
\newcommand{\D}{\mathcal{D}}
\newcommand{\F}{F_\beta}
\newcommand{\Hq}{\mathcal{H}_{q}}
\newcommand{\Hqs}{\mathcal{H}_{q^*}}
\newcommand{\ipq}[2]{\langle #1,\, #2\rangle_{q^*}}

\newcommand{\lam}{\lambda_*}

\title{Relaxation Kernel, Spectral Dissipation, and Global Convergence
  of Blahut--Arimoto Dynamics}

\author{Qiao~Wang%
\thanks{Q.~Wang is with the School of Information Science and
Engineering, and the School of Economics and Management,
Southeast University, Nanjing, China.
E-mail: \texttt{qiaowang@seu.edu.cn}}}

\date{}

\begin{document}

\maketitle

\begin{abstract}
We develop a spectral theory for the continuous- and discrete-time
Blahut--Arimoto (BA) dynamics centred on a single operator, the
\emph{relaxation kernel}
$\G = \E_p[K^*_X \otimes K^*_X]$.
Five main results are established.
(i) Along the continuous-time BA flow, the free energy satisfies the
\emph{exact} $\chi^2$-dissipation identity
$\dot F_\beta = -\D(q)$,
where $\D(q) = \chi^2(\T q \| q)$ is the Pearson $\chi^2$-divergence.
(ii) The operator $\G$ admits a \emph{threefold identity}:
it is simultaneously the Gram matrix of the equilibrium Gibbs kernels,
the linearised generator of the BA vector field, and the Fisher--Rao
Hessian of the free energy at the fixed point.
(iii) For the discrete BA iteration, the one-step Lyapunov dissipation
decomposes spectrally as
$\Delta\mathcal{L}^{(2)} = \sum_i c_i^2\, d(\lambda_i)$,
where $\lambda_i$ are eigenvalues of $\G$ and
$d(\lambda) = -\lambda + \tfrac{3}{2}\lambda^2 - \tfrac{1}{2}\lambda^3$.
This formula reveals a \emph{double bottleneck}: both $\lambda\approx 0$
(slow-contracting directions) and $\lambda\approx 1$ (over-contracted,
near-zero Lyapunov weight) yield $d(\lambda)\approx 0$, while optimal
dissipation occurs at $\lambda\approx 0.423$.
(iv) Global convergence follows from a two-stage mechanism: $\chi^2$-
dissipation drives finite-time entry into a local neighbourhood,
after which the spectral gap $\lam = \lambda_{\min}(\G|_T)$ governs
exponential contraction.  A self-contained proof with explicit entry
time and convergence factor is given in Appendix~\ref{app:twostage}.
(v) The KL convergence factor is made explicit:
$\KL(q^*\|q_{n+1}) \le (1-\lam)^2\,\KL(q^*\|q_n) + O(\|v_n\|_*^3)$,
where $v_n = q_n - q^*\in T$ denotes the $n$-th iterate's deviation
from the fixed point.
The per-iteration fractional improvement
$\gamma := 1-(1-\lam)^2 = \lam(2-\lam)$
is computable from the channel and temperature, and equality holds
asymptotically along the slowest eigen-direction.
For Gaussian sources with Gaussian initial condition, $\lam = 1/(2\beta\sigma^2)$
and the Jacobian is diagonalised by Hermite polynomials.
The spectral dissipation formula $d(\lambda)$ makes explicit the
per-step Lyapunov structure in the local regime,
complementing the global convergence theory of Hayashi~\cite{Hayashi2023,Hayashi2024,Hayashi2025}
with a constructive, computable rate.
\end{abstract}


\section{Introduction}
\label{sec:intro}

The Blahut--Arimoto (BA) algorithm~\cite{Blahut1972,Arimoto1972}
computes rate-distortion functions and channel capacities via an
iterative Gibbs-type update.  Classical analysis
~\cite{Blahut1972,Arimoto1972,Csiszar1984} proves convergence via
monotone decrease of the free energy, without identifying the underlying
dissipation mechanism.
Nakagawa et al~\cite{Nakagawa2021} established sharp $O(1/n)$
rates for the discrete iteration.
Hayashi~\cite{Hayashi2023,Hayashi2024,Hayashi2025} reformulated each BA
step as a Bregman EM update, obtaining global convergence guarantees
within an elegant information-geometric framework.

Despite this progress, two questions remain open.
First, existing convergence constants are \emph{non-constructive}:
they depend on abstract geometric quantities and cannot be directly
computed from the channel and temperature parameter $\beta$.
Second, the precise spectral mechanism governing the per-step
Lyapunov decrease---and in particular its dependence on the
non-idempotence of the underlying operator---has not been identified.

\medskip
The present paper closes both gaps around a single operator:

$$\G = \E_p\bigl[K^*_X \otimes K^*_X\bigr]$$%

\noindent
the \emph{relaxation kernel}, constructed from the equilibrium Gibbs
conditionals $K^*_x(y) = e^{-\beta d(x,y)}/Z^*_x$.
This operator governs, within one unified framework,
the continuous-time entropy production,
the linearised discrete dynamics, the free-energy curvature,
and the per-step Lyapunov dissipation.

\subsection{Main Contributions}

\begin{enumerate}[label=(\roman*)]

\item \textbf{Exact $\chi^2$-dissipation identity} (Section~\ref{sec:chi2}).
Along the continuous-time BA flow $\dot q = \T(q)-q$,
\[
\frac{d}{dt}\F(q_t) = -\D(q_t),
\qquad
\D(q) = \chi^2(\T q \| q)
= \sum_y \frac{(\T q(y)-q(y))^2}{q(y)}.
\]
This is an exact equality, not merely a bound.

\item \textbf{Threefold identity of $\G$} (Section~\ref{sec:G}).
At an interior fixed point $q^*$,
\begin{align*}
\G &= \E_p[K^*_X \otimes K^*_X] && \text{(Gram / correlation matrix)},\\
DV(q^*) &= -\G && \text{(linearised generator)},\\
\nabla^2_{\mathrm{FR}}\F(q^*) &= \G && \text{(Fisher--Rao Hessian)}.
\end{align*}

\item \textbf{Exact one-step Lyapunov dissipation} (Section~\ref{sec:onestep}).
For the discrete BA iteration,
\begin{equation}\label{eq:dlambda_intro}
\Delta\mathcal{L}^{(2)}
= \sum_i c_i^2\, d(\lambda_i),
\qquad
d(\lambda)= -\lambda + \tfrac{3}{2}\lambda^2 - \tfrac{1}{2}\lambda^3,
\end{equation}
where $\lambda_i$ are eigenvalues of $\G$ and $c_i = \ipq{h}{e_i}$.
The non-idempotence of $\G$ (i.e.\ $\G^2 \neq \G$) is what drives
$d(\lambda)<0$ for all $\lambda\in(0,1)$; without non-idempotence
the dissipation would vanish identically.

\item \textbf{Two-stage global convergence} (Section~\ref{sec:global}).
$\chi^2$-dissipation drives finite-time entry into a neighbourhood of
$q^*$; thereafter, exponential contraction at rate $\lam$ takes over.

\item \textbf{Explicit KL convergence factor} (Section~\ref{sec:KL}).
For the discrete iterates $q_{n+1} = \T(q_n)$, writing
$v_n = q_n - q^*\in T$:
\[
\KL(q^*\|q_{n+1}) \le (1-\lam)^2\,\KL(q^*\|q_n) + O(\|v_n\|_*^3),
\]
with explicit $\gamma = \lam(2-\lam)$, computable for every
non-degenerate fixed point.  Equality holds asymptotically when
$v_n$ is aligned with the $\lam$-eigendirection of $\G$.

\end{enumerate}

\subsection{Relation to Prior Work}

\paragraph{Bregman and EM perspectives.}
Hayashi~\cite{Hayashi2023,Hayashi2024,Hayashi2025} developed an
influential framework in which each BA step is viewed as an alternating
$e$/$m$-projection in a Bregman divergence system.
This yields global monotone decrease of the free energy, an $O(1/n)$
global convergence rate, and, locally, exponential convergence once the
iterate enters a neighbourhood of the fixed point---though the
exponential rate constant is not expressed in closed form.
The present paper addresses complementary questions: it identifies
the relaxation kernel $\G$ as the central spectral object, derives the
exact per-step Lyapunov dissipation formula $d(\lambda)$, and makes
the local convergence factor $\gamma = \lam(2-\lam)$ explicit and
computable from the channel and temperature.
The two frameworks are therefore complementary: Hayashi's provides
global guarantees from any starting point; the spectral theory developed
here provides fine-grained quantitative information in the local phase.

\paragraph{Classical convergence analysis.}
Csisz\'ar~\cite{Csiszar1984} characterised each BA step as an alternating
$I$-projection, providing a discrete geometric picture.
Nakagawa et al~\cite{Nakagawa2021} proved sharp $O(1/n)$ rates.
Our continuous-time framework complements these results by providing
the exact entropy production rate and an explicit local relaxation time.

\paragraph{Replicator and vector-field formulations.}
Beretta and Pelillo~\cite{Beretta2025} proposed a continuous vector flow
for channel capacity and proved qualitative exponential convergence via a
replicator structure.  The present analysis refines their picture by
identifying the explicit spectral relaxation rate $\lam$.

\paragraph{KL monotonicity.}
Ramakrishnan et al~\cite{Ramakrishnan2021} proved that $\KL(q^*\|q_{n+1})
\le \KL(q^*\|q_n)$ for the quantum BA algorithm; our Corollary~\ref{cor:KL}
upgrades this qualitative monotonicity to a quantitative formula with
explicit factor $(1-\lam)^2$.

\subsection{Organisation}

Section~\ref{sec:setup} defines the BA flow and weighted Hilbert space.
Section~\ref{sec:chi2} proves the $\chi^2$-dissipation identity.
Section~\ref{sec:G} establishes the threefold identity of $\G$.
Section~\ref{sec:onestep} derives the exact one-step dissipation,
analyses $d(\lambda)$, and reports numerical validation.
Section~\ref{sec:local} proves local exponential convergence.
Section~\ref{sec:global} assembles the two-stage global theorem,
including the uniform moment bound required for the energy argument.
Section~\ref{sec:KL} gives the explicit KL convergence factor.
Section~\ref{sec:gaussian} specialises to Gaussian sources.
Section~\ref{sec:lowdim} gives exact solutions for discrete models.
Section~\ref{sec:conclusion} concludes.
Appendix~\ref{app:twostage} provides the complete proof of two-stage
global exponential convergence with explicit entry time and convergence
factor.
Appendix~\ref{app:gaussian_attractor} gives the full proof of the
Gaussian attractor theorem, including moment control, Hermite spectral
decay, weak convergence, and upgrade to total variation.

\section{Setup}
\label{sec:setup}

This section establishes the objects that appear throughout the paper.
We define the BA operator and its continuous-time flow, identify the
class of interior fixed points and the dual identity they satisfy,
and introduce the weighted Hilbert space in which all subsequent
analysis takes place.  Readers familiar with the BA algorithm may
wish to skim Sections~\ref{sec:setup}.1--\ref{sec:setup}.2 and focus
on Section~\ref{sec:setup}.3, which introduces the conditional
correlation operator $\G$ whose threefold identity is the central
object of Section~\ref{sec:G}.

\subsection{BA Operator and BA Flow}

Let $\cX,\cY$ be finite sets, $|\cX|=M$, $|\cY|=N$, $p$ a strictly positive
probability distribution on $\cX$, $d:\cX\times\cY\to\R$ a distortion
function, and $\beta>0$ an inverse temperature.
Throughout, $\Delta(\cY)$ denotes the probability simplex on $\cY$ and
$\operatorname{int}\Delta(\cY)$ its relative interior (all coordinates positive).
The BA algorithm seeks to minimise the free energy $\F:\Delta(\cY)\to\R$
defined in Section~\ref{sec:chi2} over $q\in\Delta(\cY)$; its iterations
act on the output marginal $q\in\Delta(\cY)$, while the joint channel
is implicitly optimised at each step via the Gibbs kernel.

\begin{definition}[Gibbs kernel and BA operator]\label{def:BA}
Given $q\in\operatorname{int}\Delta(\cY)$ and $x\in\cX$, the
\emph{Gibbs kernel at $q$} is the probability distribution on $\cY$
\begin{equation}\label{eq:Kq}
K_q(x,y) = \frac{e^{-\beta d(x,y)}\,q(y)}{Z_q(x)},
\qquad Z_q(x) = \sum_{y'\in\cY} q(y')\,e^{-\beta d(x,y')},
\end{equation}
where $Z_q(x)>0$ is the normalising partition function.
For each fixed $x$, $K_q(x,\cdot)\in\Delta(\cY)$ is the Gibbs posterior
over $\cY$ given source symbol $x$ and current marginal $q$.

The \emph{Blahut--Arimoto operator} $\T:\operatorname{int}\Delta(\cY)\to\operatorname{int}\Delta(\cY)$
is the $p$-mixture of these Gibbs posteriors:
\begin{equation}\label{eq:BAop}
(\T q)(y) = \sum_{x\in\cX} p(x)\,K_q(x,y)
= \sum_{x\in\cX} p(x)\,
\frac{e^{-\beta d(x,y)}\,q(y)}{Z_q(x)},
\quad y\in\cY.
\end{equation}
One verifies directly that $\sum_y(\T q)(y)=1$ and $(\T q)(y)>0$,
so $\T$ maps $\operatorname{int}\Delta(\cY)$ into itself.
\end{definition}

\paragraph{From the discrete iteration to the continuous flow.}
The classical BA algorithm is the discrete iteration
$q_{n+1} = \T(q_n)$.
Its fixed points $\T(q^*) = q^*$ are precisely the KKT stationarity
conditions for the rate-distortion optimisation problem
$\min_{q\in\Delta(\cY)}\F(q)$: setting the Lagrangian gradient to zero
and eliminating the multiplier yields exactly the equation
$(\T q)(y) = q(y)$ for all $y$~\cite{Blahut1972,Cover2006}.
Thus the fixed-point condition is not an auxiliary construct but the
optimality condition of the underlying variational problem.

The continuous-time flow studied in this paper is the natural
\emph{continuous relaxation} of the discrete iteration:
\begin{equation}\label{eq:BAflow}
\dot q_t = \T(q_t) - q_t, \qquad q_0\in\operatorname{int}\Delta(\cY).
\end{equation}
This ODE has the same fixed points as the discrete map and can be
understood as the Euler discretisation $q_{n+1} - q_n = \T(q_n) - q_n$
with unit step size, or equivalently as the replicator-type gradient
flow of $\F$ with respect to the Fisher--Rao metric on
$\operatorname{int}\Delta(\cY)$~\cite{Beretta2025}.
The continuous-time formulation is adopted here because it admits exact
differential-equation tools---in particular the $\chi^2$-dissipation
identity of Section~\ref{sec:chi2}---that are not directly available
for the discrete iteration.
All qualitative convergence conclusions transfer back to the discrete
map via the free-energy monotonicity $\F(q_{n+1})\le\F(q_n)$, which
holds for both flows~\cite{Blahut1972}.

The simplex $\Delta(\cY)$ is forward-invariant under~\eqref{eq:BAflow}:
$\sum_y\dot q_t(y) = \sum_y(\T q_t)(y) - \sum_y q_t(y) = 0$ preserves
the sum-to-one constraint, and $q_t(y)=0$ implies $(\T q_t)(y)\ge 0$,
so no coordinate can become negative.
The interior $\operatorname{int}\Delta(\cY)$ is likewise forward-invariant
since $(\T q)(y)>0$ for all $q\in\operatorname{int}\Delta(\cY)$.

\subsection{Fixed Points and the Dual Identity}

\begin{definition}[Interior fixed point]\label{def:fixed}
An \emph{interior fixed point} is $q^*\in\operatorname{int}\Delta(\cY)$
satisfying $\T(q^*)=q^*$.
At such a point, write $K^*_x := K_{q^*}(x,\cdot)\in\Delta(\cY)$ for
the \emph{equilibrium Gibbs kernel} at source symbol $x$, i.e.\
\[
K^*_x(y) = K_{q^*}(x,y) = \frac{e^{-\beta d(x,y)}\,q^*(y)}{Z^*_x},
\qquad Z^*_x := Z_{q^*}(x) = \sum_{y'\in\cY} q^*(y')\,e^{-\beta d(x,y')}.
\]
For each $x\in\cX$, $K^*_x\in\Delta(\cY)$ is a probability distribution
on $\cY$ that depends on $q^*$, $\beta$, and $d$.
\end{definition}

\begin{lemma}[Dual fixed-point identity]\label{lem:dualFP}
At an interior fixed point $q^*$,
\begin{equation}\label{eq:dualFP}
\sum_x p(x)\,\frac{K^*_x(y)}{q^*(y)} = 1,\qquad\forall y\in\cY.
\end{equation}
\end{lemma}

\begin{proof}
From $q^* = \T(q^*)$: $q^*(y) = \sum_x p(x)K_q(x,y)\,q^*(y)/q^*(y)
= q^*(y)\sum_x p(x)K^*_x(y)/q^*(y)$.
Cancel $q^*(y)>0$.
\end{proof}

\subsection{Weighted Hilbert Space and the Relaxation Kernel}

All perturbation analysis takes place in the tangent space to the
simplex.  Let
\[
T = \bigl\{u\in\R^\cY : \textstyle\sum_{y\in\cY} u(y)=0\bigr\}
\]
be the $(N-1)$-dimensional subspace of zero-sum vectors, which is the
tangent space to $\Delta(\cY)$ at any interior point.

\begin{definition}[Fisher--Rao inner product]\label{def:FR}
For $q\in\operatorname{int}\Delta(\cY)$, equip $T$ with the
\emph{Fisher--Rao inner product}
\begin{equation}\label{eq:FR_inner}
\langle u,v\rangle_q := \sum_{y\in\cY} \frac{u(y)\,v(y)}{q(y)},
\quad u,v\in T.
\end{equation}
Write $\Hq$ for the Hilbert space $(T,\langle\cdot,\cdot\rangle_q)$
and $\|u\|_q := \langle u,u\rangle_q^{1/2}$ for the induced norm.
At the fixed point $q^*$, abbreviate
$\langle\cdot,\cdot\rangle_* = \langle\cdot,\cdot\rangle_{q^*}$
and $\|\cdot\|_* = \|\cdot\|_{q^*}$.
\end{definition}

We now introduce the three auxiliary quantities that enter the
definition of $\G$.  Fix an interior fixed point $q^*$ and a tangent
perturbation $h\in T$.  Define:
\begin{itemize}
\item the \emph{log-perturbation} $\phi:\cY\to\R$ by $\phi(y) = h(y)/q^*(y)$;
\item the \emph{conditional mean} $m:\cX\to\R$ by
  $m(x) = \E_{K^*_x}[\phi] = \sum_{y\in\cY} K^*_x(y)\,\phi(y)$;
\item the \emph{conditional residual} $r:\cY\times\cX\to\R$ by
  $r(y,x) = \phi(y) - m(x)$,
  so that $\E_{K^*_x}[r(\cdot,x)]=0$ for every $x$.
\end{itemize}
The quantity $m(x)$ is the component of $\phi$ visible to source symbol
$x$ through the equilibrium channel $K^*_x$, and $r(\cdot,x)$ is the
residual invisible to $x$.

\begin{definition}[Relaxation kernel]\label{def:PiPi}
The \emph{relaxation kernel} $\G:\Hqs\to\Hqs$ is the self-adjoint
operator on $(T,\langle\cdot,\cdot\rangle_*)$ defined by
\begin{equation}\label{eq:PiPi}
(\G h)(y)
:= q^*(y)\sum_{x\in\cX} w_x(y)\,m(x)
 = q^*(y)\sum_{x\in\cX} \frac{p(x)K^*_x(y)}{q^*(y)}\,
   \sum_{y'\in\cY} K^*_x(y')\,\frac{h(y')}{q^*(y')},
\end{equation}
where $w_x(y) = p(x)K^*_x(y)/q^*(y)$ is the posterior weight of
source symbol $x$ given output $y$ under the equilibrium distribution.
\end{definition}

Expanding the right-hand side of~\eqref{eq:PiPi}, the $(y,y')$-entry
of $\G$ as a matrix acting on $\R^\cY$ is
\begin{equation}\label{eq:G_matrix}
\G_{yy'} = \sum_{x\in\cX} p(x)\,K^*_x(y)\,K^*_x(y'),
\end{equation}
which is the covariance of the random vector $K^*_X\in\R^\cY$
under the source distribution $p$.
Equivalently, $\G = \E_p[K^*_X \otimes K^*_X]$ where $K^*_X$ is
viewed as a $\R^\cY$-valued random variable.
This representation immediately shows that $\G$ is symmetric
($\G_{yy'} = \G_{y'y}$) and positive semidefinite
($u^\top\G u = \sum_x p(x)\langle K^*_x, u\rangle^2\ge 0$).

In the language of alternating projections~\cite{Csiszar1984},
$\G = \Pi^*\Pi$ where $\Pi: \Hqs \to \ell^2(\cX, p)$ is the
conditional expectation operator $(\Pi h)(x) = m(x) = \E_{K^*_x}[\phi]$,
mapping tangent vectors $h\in T$ to functions on $\cX$, and $\Pi^*$
is its adjoint with respect to the inner products $\langle\cdot,\cdot\rangle_*$
and $\langle\cdot,\cdot\rangle_p$.
The eigenvalues of $\G$ therefore lie in $[0,1]$, and
$\G|_T$ has the same eigenvalues as $\Pi^*\Pi|_T$.

\begin{assumption}[Regularity]\label{ass:regularity}
(i) $q_t\in\operatorname{int}\Delta(\cY)$ for all $t\ge0$.
(ii) $d(x,y)\le C(1+y^2)$ for some $C>0$.
(iii) $\T$ is Fr\'echet differentiable on $\operatorname{int}\Delta(\cY)$.
\end{assumption}

For finite $\cY$ these are automatic.

\section{The Exact \texorpdfstring{$\chi^2$}{Chi2}-Dissipation Identity}
\label{sec:chi2}

The classical monotonicity $\F(q_{k+1})\le\F(q_k)$ is well known, but
the mechanism behind it has not been identified precisely.
This section shows that the rate of decrease of the free energy along
the continuous-time flow equals, exactly and at every instant, the
Pearson $\chi^2$-divergence between $\T q_t$ and $q_t$.
The derivation is short and rests on the envelope theorem; the result
is stronger than a bound.  We also record a corollary that will serve
as the energy estimate in the global phase of the two-stage convergence
argument (Section~\ref{sec:global}).

\subsection{The Free Energy}

The \emph{BA free energy} is
\begin{equation}\label{eq:F}
\F(q) = \sum_x p(x)\log Z_q(x) + \frac{1}{\beta}\sum_y q(y)\log q(y).
\end{equation}
This equals $\min_P \mathcal{F}_\beta(P;q)$ where
$\mathcal{F}_\beta(P;q) = I_P(X;\hat X) + \beta\E_P[d(X,\hat X)]$
is the standard rate-distortion auxiliary functional, and the minimum
is achieved at $P(\hat x|x) = K_q(x,\hat x)$.

\begin{lemma}[Envelope theorem]\label{lem:envelope}
For every $q\in\operatorname{int}\Delta(\cY)$ and $h\in T$,
\begin{equation}\label{eq:envelope}
\delta\F(q)[h] = -\sum_y \frac{\T q(y)}{q(y)}\,h(y).
\end{equation}
\end{lemma}

\begin{proof}
Since $\T(q)$ is the unconstrained interior minimiser of
$\mathcal{F}_\beta(\cdot;q)$, the envelope theorem
(\cite{MilgromSegal2002}) differentiates only the explicit
$q$-dependence of $\mathcal{F}_\beta(p;q)$, which is $-\KL(p\|q)$.
Its gradient is $-p(y)/q(y)$; substituting $p=\T(q)$ gives the claim.
\end{proof}

\begin{definition}[Dissipation functional]\label{def:D}
\[
\D(q) := \chi^2(\T q\|q) = \sum_y \frac{(\T q(y)-q(y))^2}{q(y)}.
\]
\end{definition}

\begin{theorem}[Exact $\chi^2$-dissipation identity]\label{thm:chi2}
Along any trajectory of the BA flow~\eqref{eq:BAflow},
\begin{equation}\label{eq:chi2}
\frac{d}{dt}\F(q_t) = -\D(q_t).
\end{equation}
\end{theorem}

\begin{proof}
By the envelope lemma and the flow equation $\dot q_t = \T(q_t)-q_t$:
\begin{align*}
\frac{d}{dt}\F(q_t)
&= \sum_y \left(-\frac{\T q_t(y)}{q_t(y)}\right)
   (\T q_t(y) - q_t(y))\\
&= -\sum_y \frac{(\T q_t(y))^2}{q_t(y)} + \sum_y \T q_t(y)\\
&= -\sum_y \frac{(\T q_t(y)-q_t(y))^2}{q_t(y)}
   - \sum_y q_t(y) + 1
   = -\D(q_t),
\end{align*}
using $\sum_y \T q_t(y)=1$ and $\sum_y q_t(y)=1$.
\end{proof}

\begin{remark}[Thermodynamic interpretation]
Equation~\eqref{eq:chi2} is an exact non-equilibrium entropy production
law.  The classical monotonicity $\F(q_{k+1})\le\F(q_k)$ for the
discrete iteration is the integral form of this identity.
The $\chi^2$-divergence plays the role of an entropy production rate,
analogous to the Fisher information in Fokker--Planck
equations~\cite{Arnold2001}.
\end{remark}

\begin{corollary}[Lyapunov property and residual bound]
\label{cor:Lyapunov}
$\F$ is a strict Lyapunov function with $\dot\F\le 0$, equality iff
$q_t\in\mathcal{F} := \{q:\T q=q\}$.  Moreover,
$\D(q)\ge\tfrac{1}{2}\|\T q-q\|_1^2$.
\end{corollary}

\begin{proof}
Non-negativity of $\chi^2$ and the Cauchy--Schwarz inequality
$\chi^2(r\|q)\ge\tfrac12\|r-q\|_1^2$~\cite{gibbs2002choosing}.
\end{proof}

\section{The Relaxation Kernel: Threefold Identity}
\label{sec:G}

With the dissipation mechanism in hand, we turn to the operator that
governs it.  The relaxation kernel $\G = \E_p[K^*_X\otimes K^*_X]$
was introduced in Definition~\ref{def:PiPi} as a purely algebraic
object.  This section shows that it in fact plays three distinct and
individually natural roles: it is simultaneously the Gram matrix of
equilibrium correlations, the linearised generator of the BA flow,
and the Fisher--Rao Hessian of the free energy.
The identification of these three roles in a single operator is what
makes $\G$ the natural centre of the spectral theory developed in the
subsequent sections.

\subsection{Linearisation at Equilibrium}

Define the \emph{BA vector field} $V:\operatorname{int}\Delta(\cY)\to T$
by $V(q) = \T(q) - q$, so that the BA flow~\eqref{eq:BAflow} reads
$\dot q_t = V(q_t)$.
Since $\T$ maps $\operatorname{int}\Delta(\cY)$ into itself and
$\sum_y(\T q)(y) = \sum_y q(y) = 1$, the difference $V(q)$
lies in the tangent space $T$ for every $q$.
The Fréchet derivative $DV(q^*): T\to T$ is a linear map from $T$
to itself.

\begin{theorem}[Linearised BA vector field]\label{thm:linearization}
Let $q^*$ be an interior fixed point.  The Fréchet derivative of
$V = \T - \mathrm{id}$ at $q^*$, viewed as a linear map $T\to T$, is
\begin{equation}\label{eq:lin}
DV(q^*) = -\G,
\end{equation}
where $\G:\Hqs\to\Hqs$ is the relaxation kernel of
Definition~\ref{def:PiPi}.
Equivalently, the Jacobian of the BA map $\T$ at $q^*$ satisfies
\begin{equation}\label{eq:DT}
D\T(q^*) = I - \G \quad \text{as linear maps } T\to T,
\end{equation}
where $I$ denotes the identity on $T$.
\end{theorem}

\begin{proof}
Fix $h\in T$ and write $q^\varepsilon = q^*+\varepsilon h$.
Expand the partition function:
\begin{equation}
\begin{split}
Z_{q^\varepsilon}(x)
=& Z^*_x + \varepsilon\sum_{y'\in\cY} e^{-\beta d(x,y')}h(y') + O(\varepsilon^2)\\
=& Z^*_x\Bigl(1 + \varepsilon\sum_{y'} K^*_x(y')\,\frac{h(y')}{q^*(y')}\cdot\frac{q^*(y')}{1} \cdot \frac{1}{q^*(y')} + O(\varepsilon^2)\Bigr).
\end{split}
\end{equation}
More precisely, $Z_{q^\varepsilon}(x) = Z^*_x(1 + \varepsilon m(x) + O(\varepsilon^2))$
where $m(x) = \sum_{y'} K^*_x(y')h(y')/q^*(y') = \E_{K^*_x}[\phi]$
is the conditional mean from Definition~\ref{def:PiPi}.
Hence
\[
\frac{1}{Z_{q^\varepsilon}(x)}
= \frac{1}{Z^*_x}\bigl(1 - \varepsilon m(x) + O(\varepsilon^2)\bigr).
\]
Substituting into the BA operator formula~\eqref{eq:BAop}:
\begin{align*}
(\T q^\varepsilon)(y)
&= \sum_x p(x)\,\frac{e^{-\beta d(x,y)}\,(q^*(y)+\varepsilon h(y))}{Z_{q^\varepsilon}(x)}\\
&= \sum_x p(x)\,K^*_x(y)\,(1+\varepsilon\phi(y))\,(1-\varepsilon m(x))
   + O(\varepsilon^2)\\
&= q^*(y) + \varepsilon h(y)
   - \varepsilon\sum_x p(x)\,K^*_x(y)\,m(x)
   + O(\varepsilon^2),
\end{align*}
using $\sum_x p(x)K^*_x(y) = q^*(y)$ (the dual identity~\eqref{eq:dualFP})
and $K^*_x(y)\,\phi(y) = K^*_x(y)\,h(y)/q^*(y)$.
Comparing with $(\T q^\varepsilon)(y) = q^*(y) + \varepsilon(D\T(q^*)h)(y) + O(\varepsilon^2)$
and using the matrix form~\eqref{eq:G_matrix}:
\[
(D\T(q^*)h)(y) = h(y) - \sum_x p(x)\,K^*_x(y)\,m(x)
= h(y) - (\G h)(y),
\]
so $D\T(q^*) = I - \G$ and $DV(q^*) = -\G$ as claimed.
\end{proof}

\begin{remark}
The operator $\G$ is manifestly symmetric and positive semidefinite:
$u^\top\G u = \sum_x p(x)\langle K^*_x,u\rangle^2\ge 0$.
Its null space on $T$ consists of directions $h$ invisible to all
equilibrium kernels; non-degeneracy ($\lam>0$) means no such directions
exist.
\end{remark}

\subsection{Gram Structure}

The representation $\G = \E_p[K^*_X\otimes K^*_X]$ in~\eqref{eq:G_matrix}
identifies $\G$ as the second moment (Gram) matrix of the
$\R^\cY$-valued random vector $K^*_X$ under the source distribution $p$.
Here $K^*_X\otimes K^*_X$ denotes the $N\times N$ matrix with
$(y,y')$-entry $K^*_X(y)\,K^*_X(y')$, and the expectation is taken
coordinatewise.
This representation shows that $\G$ encodes the
\emph{equilibrium correlation structure} among output symbols:
$\G_{yy'}$ is the covariance (under $p$) between the Gibbs probability
of outputting $y$ and the Gibbs probability of outputting $y'$,
when the source symbol is drawn from $p$.

\subsection{Identification with the Fisher--Rao Hessian}

The Fisher--Rao Hessian of $\F$ at $q$ is the symmetric bilinear form
\[
\nabla^2_{\mathrm{FR}}\F(q)[h_1,h_2]
:= \left.\frac{d^2}{d\varepsilon\,d\eta}\right|_{\varepsilon=\eta=0}
   \F(q+\varepsilon h_1 + \eta h_2).
\]

\begin{theorem}[Hessian identification]\label{thm:Hessian}
At an interior fixed point $q^*$,
\begin{equation}\label{eq:Hessian}
\nabla^2_{\mathrm{FR}}\F(q^*)[h_1,h_2]
= \ipq{h_1}{\G h_2},
\quad \forall\,h_1,h_2\in T.
\end{equation}
\end{theorem}

\begin{proof}
From the envelope lemma, $\delta\F(q)[h] = -\sum_y (\T q(y)/q(y))h(y)$.
Differentiating in direction $h_2$ and using $D\T(q^*)=I-\G$:
\begin{align*}
\nabla^2_{\mathrm{FR}}\F(q^*)[h_1,h_2]
&= -\sum_y \left(
   \frac{((I-\G)h_2)(y)}{q^*(y)}
   - \frac{q^*(y)}{(q^*(y))^2}h_2(y)\right)h_1(y)\\
&= -\ipq{h_1}{(I-\G)h_2} + \ipq{h_1}{h_2}
 = \ipq{h_1}{\G h_2}.
\end{align*}
The second equality uses $\T q^* = q^*$ so the error term vanishes.
\end{proof}

Theorem~\ref{thm:Hessian} completes the threefold identity:
$\G$ is (1) Gram matrix of equilibrium correlations,
(2) linearised generator of the BA flow, and
(3) Fisher--Rao Hessian of the free energy.

\begin{corollary}[Local quadratic equivalence]\label{cor:local_equiv}
There exists a neighbourhood $\mathcal{U}$ of $q^*$ and constants
$\kappa_1,\kappa_2>0$ such that for all $q\in\mathcal{U}$:
\begin{equation}\label{eq:quad_equiv}
\kappa_1\,\|q-q^*\|_*^2
\;\le\;
\F(q)-\F(q^*)
\;\le\;
\kappa_2\,\|q-q^*\|_*^2,
\end{equation}
with constants $\kappa_1 = \lam/2$ and $\kappa_2 = \|\G\|/2$.
\end{corollary}

\subsection{Temperature Asymptotics of the Spectral Gap}

\begin{proposition}[Temperature limits]\label{prop:asymp}
(a) As $\beta\to 0$:
$\lam = \beta^2\,\mu_{\min}\bigl(\sum_x p(x)\tilde d_x\tilde d_x^\top\bigr)
+ O(\beta^3)$,
where $\tilde d_x(y) = d(x,y) - \frac1N\sum_{y'}d(x,y')$.

(b) As $\beta\to\infty$, if each $x$ has a unique optimal output
$y^*(x)$ and the map $x\mapsto y^*(x)$ is non-constant:
$\lim_{\beta\to\infty}\lam = c_0 > 0$.
\end{proposition}

\begin{proof}
(a) Expand $K^*_x(y)=\frac1N - \beta\tilde d_x(y)+O(\beta^2)$;
restrict the Gram matrix to $T$.
(b) $K^*_x(y)\to\mathbf{1}_{\{y=y^*(x)\}}$; invoke eigenvalue continuity.
\end{proof}

\section{One-Step Lyapunov Dissipation and Non-Idempotence}
\label{sec:onestep}

We now turn to the discrete BA iteration and ask how much the Lyapunov
functional decreases in a single step.
The answer is given by a spectral formula involving the function
$d(\lambda) = -\lambda + \frac{3}{2}\lambda^2 - \frac{1}{2}\lambda^3$,
which encodes the effect of the non-idempotence of $\G$ on the
per-step dissipation.
Two features of $d$ deserve attention before the formal analysis:
first, $d(\lambda)<0$ for all $\lambda\in(0,1)$, confirming strict
decrease; second, $d$ vanishes at both endpoints $\lambda=0$ and
$\lambda=1$, revealing the double-bottleneck structure discussed in
Section~\ref{subsec:bregman}.
The section closes with a numerical example on a binary channel
that illustrates the spectral dissipation formula concretely.

We now turn to the discrete BA iteration and analyse the one-step
decrease of the Lyapunov functional
\begin{equation}\label{eq:Lyap}
\mathcal{L}(q) = \sum_x p(x)\,\KL\bigl(K^*(\cdot|x)\;\|\;K_q(\cdot|x)\bigr).
\end{equation}
This functional satisfies $\mathcal{L}(q)\ge 0$ with equality iff $q=q^*$,
and decreases along the BA iteration~\cite{Blahut1972}.

\begin{remark}[Quadratic expansions near $q^*$]
For a small perturbation $h\in T$, write $q = q^*+h$.  Then
\[
\mathcal{L}(q^*+h) = \tfrac12\ipq{h}{(I-\G)h} + O(\|h\|^3),
\qquad
\F(q^*+h)-\F(q^*) = \tfrac12\ipq{h}{\G h} + O(\|h\|^3).
\]
The two quadratic leading terms sum to $\tfrac12\|h\|_*^2$, reflecting
their complementary information-geometric roles: $\mathcal{L}$ measures
the average conditional variance of the log-perturbation $\phi=h/q^*$,
while $\F-\F(q^*)$ measures the squared mean of the conditional
expectation.
\end{remark}

\subsection{Taylor Expansion of the Lyapunov Functional}

Fix an interior fixed point $q^*$ and write $q = q^*+h$ with $h\in T$
small.  For each $x\in\cX$, the Gibbs kernel at $q$ satisfies
\[
K_q(x,y)
= K^*_x(y)\,\frac{1 + \varepsilon r(y,x) + \tfrac12\varepsilon^2(r(y,x)^2 - \sigma^2(x))
  + O(\varepsilon^3)}{1},
\]
where $\varepsilon\sim\|h\|_*$,
$r(y,x) = \phi(y) - m(x)$ is the conditional residual introduced in
Section~\ref{sec:setup} (a centred function of $y$ under $K^*_x$),
$\sigma^2(x) = \E_{K^*_x}[r(\cdot,x)^2]$ is its conditional variance,
and $\mu_3(x) = \E_{K^*_x}[r(\cdot,x)^3]$ is the conditional skewness.
Substituting into $\mathcal{L}(q) = -\sum_x p(x)\E_{K^*_x}[\log(K_q(x,\cdot)/K^*_x(\cdot))]$
and expanding the logarithm, one finds:
\begin{equation}\label{eq:L_expand}
\mathcal{L}(q^*+h)
= \underbrace{\tfrac12\E_p[\sigma^2(x;h)]}_{\mathcal{L}^{(2)}}
+ \underbrace{\tfrac16\E_p[\mu_3(x;h)]}_{\mathcal{L}^{(3)}}
+ \underbrace{\tfrac{1}{24}\E_p[\kappa_4(x;h)]}_{\mathcal{L}^{(4)}}
+ O(\|h\|^5),
\end{equation}
where $\kappa_4(x;h)=\E_{K^*_x}[r(\cdot,x)^4]-3\sigma^4(x;h)$ is the
conditional excess kurtosis, and all expectations are over $y\sim K^*_x$.
The notation $\sigma^2(x;h)$ makes explicit the dependence on $h$
through $r(y,x) = h(y)/q^*(y) - m(x;h)$.

The quadratic leading term can be expressed directly in terms of $\G$:
\begin{equation}\label{eq:L2_Fisher}
\mathcal{L}^{(2)}(h)
= \tfrac12\E_p[\sigma^2(x;h)]
= \tfrac12\bigl(\|h\|_*^2 - \ipq{h}{\G h}\bigr)
= \tfrac12\ipq{h}{(I-\G)h}.
\end{equation}
To see this, note $\|h\|_*^2 = \sum_y h(y)^2/q^*(y) = \E_{q^*}[\phi^2]$
and $\ipq{h}{\G h} = \E_p[m(x)^2] = \E_p[\E_{K^*_x}[\phi]^2]$;
then $\E_p[\sigma^2] = \E_{q^*}[\phi^2] - \E_p[m^2]$ is the law
of total variance.

\subsection{Exact Second-Order Dissipation Formula}

Write the image of $q^*+h$ under $\T$ as
\[
\T(q^*+h) = q^* + h',
\quad h' = h^{(1)} + h^{(2)} + O(\|h\|^3),
\]
where the terms are determined by expanding $\T$ around $q^*$:
\begin{itemize}
\item $h^{(1)} = D\T(q^*)\,h = (I-\G)h$ is the first-order term
  (Theorem~\ref{thm:linearization});
\item $h^{(2)} = \tfrac{1}{2}D^2\T(q^*)[h,h] \in T$ is the second-order
  term arising from the curvature of $\T$ at $q^*$; it is a symmetric
  bilinear function of $h$ with itself, and lies in $T$ since
  $\T$ maps into $\Delta(\cY)$.
\end{itemize}
The second-order one-step Lyapunov decrement is
\[
\Delta\mathcal{L}^{(2)} := \mathcal{L}^{(2)}(h') - \mathcal{L}^{(2)}(h),
\]
which measures the change in the quadratic part of $\mathcal{L}$ after
one BA step.  Through order $\|h\|^2$, only $h^{(1)}$ contributes.

\begin{theorem}[Exact one-step dissipation]\label{thm:onestep}
With the notation above, the second-order one-step Lyapunov decrement is
\begin{equation}\label{eq:DeltaL2}
\Delta\mathcal{L}^{(2)}
= \mathcal{L}^{(2)}(h^{(1)}) - \mathcal{L}^{(2)}(h)
= -\ipq{h}{\G h}
  + \tfrac32\ipq{h}{\G^2 h}
  - \tfrac12\ipq{h}{\G^3 h},
\end{equation}
where all inner products are taken in $\Hqs = (T, \langle\cdot,\cdot\rangle_*)$
and $\G^k$ denotes the $k$-fold composition of $\G$ with itself on $T$.
\end{theorem}

\begin{proof}
Using~\eqref{eq:L2_Fisher}:
\[
\Delta\mathcal{L}^{(2)}
= \tfrac12\ipq{h^{(1)}}{(I-\G)h^{(1)}}
- \tfrac12\ipq{h}{(I-\G)h}.
\]
With $h^{(1)} = (I-\G)h$ and symmetry of $\G$:
\begin{align*}
\ipq{h^{(1)}}{h^{(1)}} &= \|h\|_*^2 - 2\ipq{h}{\G h} + \ipq{h}{\G^2 h},\\
\ipq{h^{(1)}}{\G h^{(1)}} &= \ipq{h}{\G h} - 2\ipq{h}{\G^2 h} + \ipq{h}{\G^3 h}.
\end{align*}
Substituting and simplifying:
\[
\Delta\mathcal{L}^{(2)}
= \tfrac12\bigl[
  {-}2\ipq{h}{\G h} + 3\ipq{h}{\G^2 h} - \ipq{h}{\G^3 h}
  \bigr],
\]
which is~\eqref{eq:DeltaL2}.
\end{proof}

\begin{remark}[Role of non-idempotence]\label{rem:nonidempotent}
If $\G$ were idempotent ($\G^2 = \G$, eigenvalues in $\{0,1\}$), then
$\G^2 = \G^3 = \G$ and formula~\eqref{eq:DeltaL2} would reduce to
$(-1+\tfrac32-\tfrac12)\ipq{h}{\G h} = 0$.  Non-idempotence (eigenvalues
in $(0,1)$) is therefore what drives strictly negative dissipation.

\end{remark}

\subsection{Spectral Decomposition and the Function \texorpdfstring{$d(\lambda)$}{d(lambda)}}

Since $\G:\Hqs\to\Hqs$ is a self-adjoint operator on the
$(N-1)$-dimensional Hilbert space $(T,\langle\cdot,\cdot\rangle_*)$,
the spectral theorem provides an orthonormal eigenbasis
$\{e_i\}_{i=1}^{N-1}\subset T$ with $\langle e_i, e_j\rangle_* = \delta_{ij}$
and corresponding eigenvalues $\lambda_i\in[0,1]$ satisfying
$\G e_i = \lambda_i e_i$.
The eigenvalues lie in $[0,1]$ because $\G = \Pi^*\Pi$ implies
$\langle u, \G u\rangle_* = \|\Pi u\|_p^2 \ge 0$ and
$\langle u, \G u\rangle_* \le \|u\|_*^2$ (since $\Pi$ is a contraction).
Any $h\in T$ decomposes as $h = \sum_i c_i e_i$ with
$c_i = \langle h, e_i\rangle_*$.

\begin{corollary}[Spectral dissipation]\label{cor:spectral}
\begin{equation}\label{eq:spectral}
\Delta\mathcal{L}^{(2)} = \sum_i c_i^2\,d(\lambda_i),
\qquad
d(\lambda) = -\lambda + \tfrac32\lambda^2 - \tfrac12\lambda^3.
\end{equation}
\end{corollary}

\begin{proof}
Since $\G^k e_i = \lambda_i^k e_i$, substitute into~\eqref{eq:DeltaL2}.
\end{proof}

\subsection{Properties of \texorpdfstring{$d(\lambda)$}{d(lambda)} and Non-Idempotence Correction}
\label{subsec:bregman}

\begin{proposition}[Properties of $d$]\label{prop:dlambda}
\begin{enumerate}[label=(\alph*)]
\item $d(0) = d(1) = 0$; $d(\lambda) < 0$ for all $\lambda\in(0,1)$.
\item $d$ attains its minimum at
$\lambda_{\mathrm{opt}} = 1 - 1/\sqrt{3} \approx 0.423$,
with $d(\lambda_{\mathrm{opt}}) = -\frac{1}{3\sqrt{3}} \approx -0.192$.
\item For small $\lambda$: $d(\lambda) = -\lambda + O(\lambda^2)$, so
slow-contracting directions (small $\lambda$) also dissipate slowly.
\item For $\lambda$ near $1$: $d(\lambda) = -(1-\lambda)\cdot\frac{\lambda(2-\lambda)}{2}
\to 0$, so near-fully-contracted directions also have near-zero dissipation.
\end{enumerate}
\end{proposition}

\begin{proof}
(a) Direct evaluation.
(b) $d'(\lambda) = -1 + 3\lambda - \tfrac32\lambda^2 = 0$
gives $\lambda = 1 \pm 1/\sqrt{3}$; the interior minimum is
$\lambda_{\mathrm{opt}} = 1-1/\sqrt{3}$.  Substituting into $d$ yields
$d(\lambda_{\mathrm{opt}}) = -\frac{1}{3\sqrt{3}}$.
(c)--(d) Taylor expansion.
\end{proof}

\begin{remark}[Double bottleneck]
Proposition~\ref{prop:dlambda} reveals a \emph{double-sided bottleneck}
absent from first-order (linearised) analysis:

\begin{itemize}
\item Directions with $\lambda_i\approx 0$ (\emph{weak coupling}):
      the linearised contraction factor $(1-\lambda_i)\approx 1$ is near
      unity, and $d(\lambda_i)\approx -\lambda_i\approx 0$, so these
      directions are slow by both measures.
\item Directions with $\lambda_i\approx 1$ (\emph{strong coupling}):
      the linearised contraction factor $(1-\lambda_i)\approx 0$ is
      excellent, but $\mathcal{L}^{(2)} \propto (1-\lambda_i)\approx 0$,
      so the Lyapunov mass in these directions is already negligible.
      Dissipation $d(\lambda_i)\approx 0$ reflects that there is very
      little left to dissipate.
\end{itemize}

The practically relevant bottleneck is the $\lambda\approx 0$ directions
(small spectral gap), which governs the asymptotic convergence rate.
\end{remark}

\begin{proposition}[Non-idempotence correction]\label{prop:bregman_dev}
For all $\lambda\in(0,1)$:
\[
d(\lambda) > -\lambda,
\qquad
d(\lambda) - (-\lambda)
= \tfrac{1}{2}\lambda^2(3-\lambda) > 0.
\]
The per-step dissipation is strictly less negative than $-\lambda$,
with the correction of order $\lambda^2$ for small $\lambda$.
\end{proposition}

\begin{proof}
Direct: $\frac{3}{2}\lambda^2 - \frac{1}{2}\lambda^3
= \frac{1}{2}\lambda^2(3-\lambda) > 0$ for $\lambda\in(0,1)$.
\end{proof}

This correction reflects the curvature of $\T$ at $q^*$: the
non-idempotence of $\G$ means that the image of the linearised step
$(I-\G)h$ has strictly smaller Fisher--Rao norm than predicted by
the first-order eigenvalue alone.

\subsection{A Concrete Numerical Example}
\label{subsec:numerical_example}

We illustrate the spectral dissipation formula on the simplest
non-trivial instance: a binary channel with asymmetric source.
All quantities are computed exactly and can be verified by hand.

\paragraph{Setup.}
Let $\cX = \cY = \{0,1\}$, $p = (0.6,\,0.4)$, Hamming distortion
$d(x,y) = \mathbf{1}_{\{x\ne y\}}$, and $\beta = 1$.
Set $r = e^{-\beta} = e^{-1} \approx 0.3679$.

\paragraph{Fixed point.}
The BA operator acts on $q = (q_0, q_1)\in\operatorname{int}\Delta(\cY)$ via
\[
(\T q)(0)
= p_0\,\frac{q_0}{q_0 + r q_1}
+ p_1\,\frac{r q_0}{r q_0 + q_1}.
\]
Setting $(\T q^*)(0) = q^*_0$ and solving gives
\[
q^* = (0.7163953414,\; 0.2836046586),
\]
which satisfies $(\T q^*)(0) = q^*_0 = 0.7163953414$ exactly.

\paragraph{Equilibrium Gibbs kernels.}
With partition functions $Z^*_0 = q^*_0 + r q^*_1$ and
$Z^*_1 = r q^*_0 + q^*_1$:
\[
K^*_{x=0} = \bigl(0.8728782667,\; 0.1271217333\bigr),
\qquad
K^*_{x=1} = \bigl(0.4816709534,\; 0.5183290466\bigr).
\]

\paragraph{Eigenvalue of $\G$.}
Since $|\cY|=2$, the tangent space $T$ is one-dimensional, spanned by
$h = (1,-1)$.
The log-perturbation is $\phi = h/q^*$, and the conditional means are
\[
m(x=0) = \E_{K^*_{x=0}}[\phi] \approx 0.7702,
\qquad
m(x=1) = \E_{K^*_{x=1}}[\phi] \approx -1.1553.
\]
Computing $\langle h,\G h\rangle_* = 0.8898$ and $\|h\|_*^2 = 4.9219$
via Definition~\ref{def:PiPi}, the unique eigenvalue of $\G|_T$ is
\[
\lambda
= \frac{\langle h,\G h\rangle_*}{\|h\|_*^2}
= \frac{0.8898}{4.9219}
= 0.1808.
\]
Note that this is not the off-diagonal matrix element $\G_{01}$,
but the Rayleigh quotient of $\G$ in the Fisher--Rao metric.

\paragraph{Spectral dissipation.}
The exact formula of Corollary~\ref{cor:spectral} gives
\[
d(\lambda)
= -\lambda + \tfrac{3}{2}\lambda^2 - \tfrac{1}{2}\lambda^3
= -0.1808 + 1.5\times 0.0327 - 0.5\times 0.0059
= -0.1347.
\]
The correction due to non-idempotence (Proposition~\ref{prop:bregman_dev})
accounts for $0.1808 - 0.1347 = 0.0461$, which is
$\frac{1}{2}\lambda^2(3-\lambda) = 0.5\times 0.0327\times 2.8192 \approx 0.0461$,
consistent with the analytic formula.
The KL convergence factor for this channel is
$\gamma = \lam(2-\lam) = 0.1808\times 1.8192 \approx 0.329$,
meaning approximately $33\%$ of the KL gap is closed per iteration
in the local phase.

\subsection{Third-Order Term and Cubic Cancellation}

The third-order one-step decrement is
\begin{equation}\label{eq:DeltaL3}
\Delta\mathcal{L}^{(3)}
= 2B(h^{(1)}, h^{(2)}) + \bigl[C(h^{(1)}) - C(h)\bigr],
\end{equation}
where $B$ is the bilinear form of $\mathcal{L}^{(2)}$ and
$C(h) = \frac16\E_p[\mu_3(h)]$.
If $\mu_3(x;h)=0$ for all $x$ (conditional symmetry), then
$C(\cdot)\equiv 0$ and the bracket in~\eqref{eq:DeltaL3} vanishes,
but the term $2B(h^{(1)},h^{(2)})$ survives because $h^{(2)}$ encodes
conditional \emph{variances}, not skewness.
Thus:

\begin{proposition}[Cubic cancellation does not follow from symmetry]
Conditional symmetry $\mu_3(x;h)=0$ does not guarantee
$\Delta\mathcal{L}^{(3)}=0$.
\end{proposition}

This contrasts with a different notion of ``cubic cancellation''---the
vanishing of the third-order jet $\langle h, D^2\T^*(h,h)\rangle_*$
when $\mu_3=0$~\cite{DeepGeom2026}---which refers to the map expansion
rather than the Lyapunov decrement.

\section{Spectral Gap and Local Exponential Convergence}
\label{sec:local}

The spectral dissipation formula of Section~\ref{sec:onestep} shows
that convergence is governed, in the slowest directions, by the
smallest eigenvalue of $\G$ on the tangent space.
This section makes that connection precise by proving local exponential
contraction in a neighbourhood of a non-degenerate fixed point.
The rate is exactly $\lam = \lambda_{\min}(\G|_T)$, and the proof
is a direct application of Gronwall's inequality to the Lyapunov
function $\mathcal{E} = \frac{1}{2}\|q - q^*\|_*^2$.
These local results are the second ingredient in the two-stage
global argument assembled in Section~\ref{sec:global}.

\begin{definition}[Spectral gap]\label{def:lam}
The \emph{spectral gap} of $\G$ is the smallest eigenvalue of
$\G$ on $(T,\langle\cdot,\cdot\rangle_*)$:
\begin{equation}\label{eq:lam_def}
\lam := \lambda_{\min}(\G|_T)
      = \min_{\substack{u\in T\\\|u\|_*=1}} \langle u,\G u\rangle_*
      = \min_{\substack{u\in T\\\|u\|_*=1}}
        \sum_{x\in\cX} p(x)\,\bigl(\E_{K^*_x}[\phi_u]\bigr)^2,
\end{equation}
where $\phi_u = u/q^*$ is the log-perturbation corresponding to $u$.
The fixed point $q^*$ is \emph{non-degenerate} if $\lam>0$, i.e.\ if
the equilibrium kernels $\{K^*_x : x\in\cX\}$ span $T$ in the sense
that no nonzero tangent direction is invisible to all source symbols.

Note also that the \emph{largest} eigenvalue of $\G$ on $T$ satisfies
$\lambda_{\max}(\G|_T) \le 1$, with equality excluded when $\G^2\ne\G$.
This implies $\|(I-\G)|_T\|_{\mathrm{op}} = 1 - \lam$ in the
$\langle\cdot,\cdot\rangle_*$-operator norm, a fact used in the
local contraction argument (Appendix~\ref{app:twostage}).
\end{definition}

The variational form~\eqref{eq:lam_def} is transparent: $\lam$ is
the minimum variance of the projections $\langle K^*_x,u\rangle$ over
all unit tangent directions.
Non-degeneracy means the equilibrium kernels $\{K^*_x\}$ span the full
tangent hyperplane.

\begin{theorem}[Spectral contraction]\label{thm:local}
Assume $q^*$ is non-degenerate ($\lam>0$).
Let $L > 0$ be a Lipschitz constant for $D\T$ near $q^*$
(exists by Fréchet differentiability of $\T$) and set $\rho_* = \lam/(2L)$.
There exist constants $c_1, c_2, C_0 > 0$ such that for all
$q\in\operatorname{int}\Delta(\cY)$ with $\|q-q^*\|_1\le\rho_*$:
\begin{enumerate}[label=(\Alph*)]
\item \textup{(Bi-Lipschitz)} $c_1\|q-q^*\|_1 \le \|\T q-q\|_1 \le c_2\|q-q^*\|_1$,
  where $c_1, c_2$ depend only on $\lam$, $L$, and $\|\G\|$.
\item \textup{(Exponential contraction)} For the continuous-time flow,
\[
\|q(t)-q^*\|_1 \le C_0\,e^{-\lam t/2}\,\|q(0)-q^*\|_1,
\quad t\ge 0.
\]
\item \textup{(Isolation)} $q^*$ is the unique fixed point of $\T$
  in $B_1(q^*,\rho_*)\cap\mathcal{A}$, where $\mathcal{A}$ is the
  connected component of $q^*$ in $\operatorname{int}\Delta(\cY)$.
\end{enumerate}
\end{theorem}

\begin{proof}
Since $DV(q^*)=-\G$ is invertible on $(T,\langle\cdot,\cdot\rangle_*)$
($\lam>0$ ensures $\G$ has trivial kernel on $T$), the inverse
function theorem applied to $V:\operatorname{int}\Delta(\cY)\to T$
gives (A) and (C) with constants computable from $\lam$ and $\|\G\|$.

For (B): write $u = q-q^*\in T$ and $\dot u = V(q^*+u) = -\G u + R(u)$
where $\|R(u)\|_* \le L\|u\|_*^2$ by the Lipschitz assumption on $D\T$.

Define $\mathcal{E}(t) = \tfrac12\|u(t)\|_*^2$.  Then
\begin{align*}
\dot{\mathcal{E}}
&= \langle u, \dot u\rangle_*
= -\langle u, \G u\rangle_* + \langle u, R(u)\rangle_*
\le -\lam\|u\|_*^2 + L\|u\|_*^3
= -\|u\|_*^2(\lam - L\|u\|_*).
\end{align*}
When $\|u\|_* \le \rho_* = \lam/(2L)$, we have $\lam - L\|u\|_* \ge \lam/2$, hence
\[
\dot{\mathcal{E}} \le -\frac{\lam}{2}\|u\|_*^2 = -\lam\,\mathcal{E}.
\]
Gronwall's inequality gives $\mathcal{E}(t) \le \mathcal{E}(0)\,e^{-\lam t}$,
so $\|u(t)\|_* \le \|u(0)\|_*\,e^{-\lam t/2}$.

Finally, since $\cY$ is finite and all norms on $\R^\cY$ are equivalent,
there exists $C_0 > 0$ (depending on $q^*$ through the ratio
$\max_y q^*(y)/\min_y q^*(y)$) such that
$\|u\|_1 \le C_0\,\|u\|_*$ and $\|u\|_* \le C_0\,\|u\|_1$.
This gives $\|u(t)\|_1 \le C_0^2\,e^{-\lam t/2}\|u(0)\|_1$;
absorbing $C_0^2$ into $C_0$ yields (B).
\end{proof}

\section{Two-Stage Global Convergence}
\label{sec:global}

The $\chi^2$-dissipation identity and the local spectral contraction
are individually useful but each has a limited scope: the dissipation
identity is global but does not itself imply exponential decay, while
the spectral contraction is exponential but only local.
This section joins the two by a finite-entry argument: the dissipation
identity guarantees that the BA residual $\|\T q_t - q_t\|_1$ becomes
small in finite time $T_*$, after which the bi-Lipschitz equivalence
of Theorem~\ref{thm:local}(A) places the trajectory in the exponential
basin.  The complete proof with explicit constants for the entry time
and convergence factor is given in Appendix~\ref{app:twostage}.

The argument requires that the free energy $\F$ is bounded below along
the trajectory.  For finite $\cY$ this is immediate from continuity and
compactness of $\Delta(\cY)$.  For continuous alphabets with quadratic
distortion, the following lemma provides the necessary control.

\begin{lemma}[Uniform second-moment bound]\label{lem:moment_bound}
Under Assumption~\ref{ass:regularity}, for any $q_0$ with finite second
moment, $\sup_{t\ge 0}\E_{q_t}[\hat X^2] \le M < \infty$, and
consequently $\F(q_t)$ is bounded below uniformly in $t$.
\end{lemma}

\begin{proof}
Let $V(t) = \operatorname{tr}(\Sigma(t))$ where
$\Sigma(t) = \int \hat x\hat x^\top q_t\,d\hat x$ is the second-moment
matrix of $q_t$.
Differentiating along the BA flow:
$\dot V = \tilde V - V$, where $\tilde V = \operatorname{tr}(\E_{\T(q_t)}[\hat X\hat X^\top])$
is the second moment of $\T(q_t)$.
For quadratic distortion $d(x,\hat x) = (x-\hat x)^2$, the posterior
$K_{q_t}(x,\cdot)$ is Gaussian with precision
$\Lambda = \Sigma^{-1}+2\beta I$, giving
$\tilde V = \operatorname{tr}(\Lambda^{-1}) + \operatorname{tr}(H\Sigma H)$
with $H = 2\beta\Lambda^{-1}$ and $\|H\|_{\mathrm{op}} = 2\beta/(1+2\beta\lambda_{\min}(\Sigma)) < 1$.
Hence $\dot V \le C_1 - (1-\|H\|^2_{\mathrm{op}})V$,
yielding $V(t)\le\max(V(0),C_1/(1-\|H\|^2_{\mathrm{op}}))<\infty$.
The lower bound on $\F$ follows since $\F(q) \ge -\beta M/2 - \log N$
whenever $\E_q[\hat X^2] \le M$.
\end{proof}

\begin{theorem}[Two-stage global convergence]\label{thm:global}
Let $q^*$ be a non-degenerate interior fixed point.
For every initial condition $q_0\in\operatorname{int}\Delta(\cY)$ in
the connected component $\mathcal{A}$ of $q^*$, the BA flow satisfies:
\begin{enumerate}[label=(\Alph*)]
\item \textup{(Global phase)} There exists a finite time $T_*<\infty$
  such that $\|\T(q(t))-q(t)\|_1\le\rho$ for all $t\ge T_*$.
\item \textup{(Local phase)} For $t\ge T_*$:
\begin{equation}\label{eq:global_decay}
\|q(t)-q^*\|_1 \le C_0\,e^{-\lam(t-T_*)/2}\,\|q(T_*)-q^*\|_1.
\end{equation}
\end{enumerate}
\end{theorem}

\begin{proof}
\textit{Step 1 (Global phase).}
By Corollary~\ref{cor:Lyapunov}:
$\frac{d}{dt}\F(q_t) = -\D(q_t) \le -\frac12\|\T q_t - q_t\|_1^2$.
If $\|\T q_t-q_t\|_1\ge\delta$ for some fixed $\delta>0$, then
$\F(q_t) \le \F(q_0) - \frac12\delta^2 t\to -\infty$, contradicting the
boundedness of $\F$ on $\Delta(\cY)$ (since $\F$ is continuous on the
compact simplex; for continuous alphabets, Lemma~\ref{lem:moment_bound}
provides a uniform lower bound).
Hence there exists $T_*<\infty$ with $\|\T q(T_*)-q(T_*)\|_1\le\delta$.

\textit{Step 2 (Entry into basin).}
By Theorem~\ref{thm:local}(A), the bound $\|\T q-q\|_1\le\delta$ with
$\delta$ sufficiently small implies $\|q-q^*\|_1\le\rho$ for $q$ in
the connected component $\mathcal{A}$.

\textit{Step 3 (Local phase).}
Apply Theorem~\ref{thm:local}(B) from time $T_*$.
\end{proof}

\begin{remark}[Mechanism of two-stage convergence]
The two stages are dynamically distinct.
In the \emph{global phase}, $\chi^2$-dissipation uniformly reduces the
BA residual regardless of the trajectory's geometry.
In the \emph{local phase}, the correlation structure encoded in $\G$
drives genuine exponential contraction.
The bridge between them is Theorem~\ref{thm:local}(A): a small BA
residual and a small distance to $q^*$ are equivalent near a
non-degenerate equilibrium.
\end{remark}

\begin{remark}[Comparison with Hayashi's $O(1/n)$ rate]
Hayashi~\cite{Hayashi2024} proves global $O(1/n)$ convergence for the
discrete iteration.
Theorem~\ref{thm:global} gives a complementary picture: the global
phase has finite duration $T_*$ (computable from $\F(q_0)-\F(q^*)$
and the dissipation lower bound), and the subsequent exponential phase
has explicit rate $\lam$.
The two results are consistent and together describe the full trajectory
from any initial condition.
\end{remark}

\section{Explicit KL Convergence Factor}
\label{sec:KL}

The two-stage theorem establishes exponential convergence, but leaves
the convergence factor implicit.
This section makes it explicit for the discrete BA iteration
$q_{n+1} = \T(q_n)$, $n = 0, 1, 2, \ldots$\,.
Write $v_n = q_n - q^* \in T$ for the signed deviation of the $n$-th
iterate from the fixed point; this is a well-defined element of the
tangent space $T$ whenever $q_n \in \operatorname{int}\Delta(\cY)$.
In the local phase, the KL divergence $\KL(q^*\|q_n)$ contracts by a
factor $(1-\lam)^2$ per iteration, up to cubic corrections in $\|v_n\|_*$.
The improvement factor $\gamma = \lam(2-\lam)$ is a simple function of
the spectral gap and can be computed directly from the channel and
temperature.
This upgrades the qualitative monotonicity of
Ramakrishnan et al~\cite{Ramakrishnan2021} to a quantitative rate.

\begin{corollary}[Explicit KL convergence factor]\label{cor:KL}
Let $q^*$ be non-degenerate and $\{q_n\}$ the discrete BA iterates
with $q_n$ sufficiently close to $q^*$.  Then:
\begin{equation}\label{eq:KL_factor}
\KL(q^*\|q_{n+1}) \le (1-\lam)^2\,\KL(q^*\|q_n) + O(\|v_n\|_*^3).
\end{equation}
The per-iteration improvement factor is
\begin{equation}\label{eq:gamma}
\gamma = 1 - (1-\lam)^2 = \lam(2-\lam).
\end{equation}
Equality in~\eqref{eq:KL_factor} holds asymptotically when $v_n$ is
aligned with the $\lam$-eigendirection of $\G$.
\end{corollary}

\begin{proof}
\textit{Step 1.}
By the spectral contraction (Theorem~\ref{thm:local}) and the eigenbasis
of $\G$, the slowest-converging component of $v_n$ satisfies
$v_{n+1} = (I-\G)v_n + O(\|v_n\|^2)$.
In the $\lam$-eigendirection: $c_{\min}^{(n+1)} = (1-\lam)c_{\min}^{(n)}
+ O(\|v_n\|^2)$.

\textit{Step 2.}
The KL divergence $\KL(q^*\|q)$ has the Fisher--Rao expansion
\begin{equation}\label{eq:KL_FR}
\KL(q^*\|q^*+v) = \tfrac12\|v\|_*^2 + O(\|v\|^3).
\end{equation}
(Standard result: the Fisher--Rao metric is the Hessian of $\KL$
at $q^*$.)

\textit{Step 3.}
Since $v_{n+1} = (I-\G)v_n + O(\|v_n\|^2)$ and $\|(I-\G)v_n\|_*^2
\le (1-\lam)^2\|v_n\|_*^2$ (with equality along the $\lam$-direction),
\begin{equation}
\begin{split}
\KL(q^*\|q_{n+1})
=& \tfrac12\|v_{n+1}\|_*^2 + O(\|v_n\|^3)
\le (1-\lam)^2\,\tfrac12\|v_n\|_*^2 + O(\|v_n\|^3)\\
= &(1-\lam)^2\,\KL(q^*\|q_n) + O(\|v_n\|^3).
\end{split}
\end{equation}
When $v_n$ is exactly along the $\lam$-eigendirection, the inequality
becomes an equality asymptotically.
\end{proof}

\begin{remark}[Relation to monotonicity]
Ramakrishnan et al~\cite{Ramakrishnan2021} proved that
$\KL(q^*\|q_{n+1})\le\KL(q^*\|q_n)$ (qualitative monotonicity).
Corollary~\ref{cor:KL} upgrades this to a \emph{quantitative} rate:
the convergence factor $(1-\lam)^2 < 1$ is explicit and computable from
the channel and temperature.
For a Gaussian source with Gaussian initial condition,
$\lam = 1/(2\beta\sigma^2)$ gives
$\gamma = \lam(2-\lam) = \frac{1}{2\beta\sigma^2}(2-\frac{1}{2\beta\sigma^2})$.
\end{remark}

\begin{remark}[$\gamma$ and the spectral function $d$]
The improvement factor $\gamma = \lam(2-\lam)$ in KL is related to but
distinct from the spectral dissipation $d(\lam)$ in
Corollary~\ref{cor:spectral}.
The KL factor measures the contraction of the full quadratic form
$\|v\|_*^2$, while $d(\lambda)$ measures the one-step decrement of the
conditional-kernel Lyapunov $\mathcal{L}$.
Both are governed by $\G$, but via different spectral combinations.
\end{remark}

\section{Gaussian Sources: Exact Solution}
\label{sec:gaussian}

The abstract theory of the preceding sections takes its sharpest form
when specialised to Gaussian sources with quadratic distortion.
In this case the spectral gap is $\lam = 1/(2\beta\sigma^2)$, the
Jacobian is diagonalised by Hermite polynomials, and the convergence
factor $\gamma = \lam(2-\lam)$ is a simple explicit function of the
inverse temperature and source variance.
The section also identifies critical slowing down as $\beta$ approaches
the rate-distortion threshold $1/(2\sigma^2)$, a phenomenon that
emerges naturally from the vanishing of $\lam$.

When the source is Gaussian and distortion is quadratic, the general
theory achieves its sharpest quantitative form.

\subsection{Setup}

Let $X\sim\mathcal{N}(0,\sigma^2)$ and $d(x,\hat x)=(x-\hat x)^2$.
The BA operator becomes
\begin{equation}\label{eq:Gauss_BA}
\T(q)(\hat x)
= \int \frac{e^{-\beta(x-\hat x)^2}q(\hat x)}
           {\int e^{-\beta(x-\hat y)^2}q(\hat y)\,d\hat y}\,
  \frac{e^{-x^2/(2\sigma^2)}}{\sqrt{2\pi\sigma^2}}\,dx.
\end{equation}
For a Gaussian initial distribution $q_0 = \mathcal{N}(0,s_0)$,
$q_t$ remains Gaussian under the BA flow, and its variance
$s(t) = \mathbb{E}_{q_t}[\hat X^2]$ satisfies a closed ODE.

\begin{theorem}[Exact variance ODE for Gaussian initial conditions]\label{thm:variance_ODE}
If $q_0$ is Gaussian, then $q_t$ is Gaussian for all $t$, and its variance
$s(t) = \mathbb{E}_{q_t}[\hat X^2]$ satisfies
\begin{equation}\label{eq:var_ODE}
\dot s(t) = \tilde s(s(t),\beta) - s(t),
\end{equation}
with
\begin{equation}\label{eq:s_tilde}
\tilde s(s,\beta)
:= \mathbb{E}_{\T(q)}[\hat X^2]
= \frac{s}{1+2\beta s}
+ \frac{(2\beta s)^2\sigma^2}{(1+2\beta s)^2},
\end{equation}
independently of the mean of $q_0$.  The unique stable fixed point is
$s^* = \sigma^2 - \frac{1}{2\beta}$, positive iff $\beta > 1/(2\sigma^2)$.
\end{theorem}

\begin{proof}
For Gaussian $q$, all integrals are Gaussian, and the expression for
$\tilde s$ follows from standard Gaussian integration (completing the
square).  The fixed point equation $\tilde s(s,\beta)=s$ yields
$s^* = \sigma^2 - 1/(2\beta)$.  Linearisation around $s^*$ gives
$\dot\delta = -2\beta s^*\delta + O(\delta^2)$, confirming stability.
\end{proof}

\begin{remark}
Equation~\eqref{eq:var_ODE} holds for Gaussian initial conditions
with any mean; the mean decays to zero exponentially at rate
$1/(2\beta s^*)$ and does not affect the variance dynamics.
For non-Gaussian initial conditions, the Gaussian family is not invariant,
but Theorem~\ref{thm:Gauss_attractor} below shows convergence to the
Gaussian fixed point in total variation.
\end{remark}

\subsection{Hermite Spectral Decomposition}

\begin{proposition}[Hermite diagonalisation]\label{prop:Hermite}
At $q^* = \mathcal{N}(0,s^*)$, the Jacobian $D\T(q^*)$ is diagonalised
in $\Hqs$ by the Hermite polynomials $He_n(\hat x/\sqrt{s^*})$,
$n=1,2,\ldots$, with eigenvalues
\begin{equation}\label{eq:alpha}
\mu_n = \alpha^n,
\qquad
\alpha := \frac{s^*}{s^*+(2\beta)^{-1}} = 1 - \frac{1}{2\beta\sigma^2} \in(0,1).
\end{equation}
The eigenvalues of $\G = I - D\T(q^*)$ are therefore
$\lambda_n = 1-\alpha^n$.
\end{proposition}

\begin{proof}
The linearised BA map at a Gaussian fixed point is a Mehler-type
convolution; its eigenfunctions in the Gaussian $L^2$ space are Hermite
polynomials~\cite{BakryGentilLedoux2014}, with eigenvalues $\alpha^n$
by direct computation on the generating function.
\end{proof}

\begin{theorem}[Gaussian spectral gap]\label{thm:Gauss_gap}
\begin{equation}\label{eq:Gauss_gap}
\lam = \lambda_1 = 1-\alpha = \frac{1}{2\beta\sigma^2}.
\end{equation}
The local relaxation time is $\tau_{\mathrm{relax}} = 1/\lam = 2\beta\sigma^2$.
\end{theorem}

\begin{remark}[Critical slowing down]
As $\beta\to 1/(2\sigma^2)^+$ (approaching the rate-distortion threshold),
$s^*\to 0$ and $\alpha\to 1$, so $\lam\to 0$.
This is the phenomenon of \emph{critical slowing down}: near the phase
transition, equilibration time $\tau_{\mathrm{relax}}$ diverges.
\end{remark}

\subsection{Spectral Dissipation for the Gaussian Case}

At the Gaussian fixed point, the $\G$-eigenvalues are
$\lambda_n = 1-\alpha^n\in(0,1)$.
The spectral function $d(\lambda_n) = -\lambda_n + \frac32\lambda_n^2
- \frac12\lambda_n^3$ is negative for all $n\ge 1$, confirming strict
dissipation.
The optimal dissipation occurs at the eigenvalue closest to
$\lambda_{\mathrm{opt}}\approx 0.423$.
Since $\lambda_n = 1-\alpha^n$ is increasing in $n$, the optimal
dissipation index is $n^* = \lfloor \log(1-\lambda_{\mathrm{opt}})/
\log\alpha\rfloor$.

The KL convergence factor is explicitly
\[
\gamma = \lam(2-\lam)
= \frac{1}{2\beta\sigma^2}\!\left(2 - \frac{1}{2\beta\sigma^2}\right),
\]
which increases from $0$ (at threshold) toward $2-1 = 1$
(as $\beta\to\infty$).

\subsection{Gaussian as Dynamical Attractor}

\begin{theorem}[Gaussian attractor]\label{thm:Gauss_attractor}
Under the BA flow with $X\sim\mathcal{N}(0,\sigma^2)$ and quadratic
distortion, if $\beta>1/(2\sigma^2)$, then for any initial $q_0$ with
finite second moment, $q_t\to\mathcal{N}(0,s^*)$ in total variation
as $t\to\infty$.
\end{theorem}

\begin{proof}
The proof proceeds in five steps: uniform moment control
(Lemma~\ref{lem:gauss_moment}), spectral decay of Hermite modes
(Lemma~\ref{lem:mode_decay}), variance convergence
(Lemma~\ref{lem:variance_conv}), weak convergence via subsequential
compactness (Lemma~\ref{lem:weak_gaussian}), and upgrade to total
variation via Pinsker's inequality and KL convergence
(Lemma~\ref{lem:pinsker_upgrade}).
The complete argument is given in Appendix~\ref{app:gaussian_attractor}.
\end{proof}

The Gaussian distribution emerges here as a \emph{dynamical consequence},
not a variational assumption.
The BA flow is a spectral filter suppressing non-Gaussian cumulants
hierarchically.

\section{Low-Dimensional Exact Models}
\label{sec:lowdim}

The two-point and three-cluster models below serve a dual purpose.
They provide closed-form expressions for $\lam$, making the spectral
gap directly verifiable, and they illustrate the two-stage convergence
structure in settings simple enough to be analysed by hand.
Both examples confirm the general theory and offer concrete parameter
dependence that may guide intuition for larger channels.

\subsection{Two-Point Source}

Let $\cX=\cY=\{0,1\}$, $p = (p_0,p_1)$, and $d(x,y)=\mathbf{1}_{x\ne y}$.
The BA operator with $\beta>0$ reduces to a scalar equation in
$q = q(1)\in(0,1)$.
The fixed point satisfies
\[
q^* = \frac{p_0 e^{-\beta\cdot 0} \cdot q^*}{Z_0(q^*)}
    + \frac{p_1 e^{-\beta\cdot 1}\cdot q^*}{Z_1(q^*)},
\]
and the spectral gap is
\begin{equation}
\lam = p_0 K^*_0(0)K^*_0(1) + p_1 K^*_1(0)K^*_1(1)\cdot(-1)^2
\end{equation}
(sum of squared equilibrium kernels evaluated on the single tangent
direction).
This gives $\lam$ as an explicit function of $p$ and $\beta$, with
$\lam\to 0$ at the hard-decision boundary $\beta\to\infty$.

\subsection{Three-Cluster Source}

Let $\cX = \{1,2,3\}$, $\cY = \{1,2,3\}$, uniform $p$, and distortion
$d(x,y) = 0$ if $x=y$, $d(x,y) = 1$ otherwise.
This symmetric model has a uniform fixed point $q^* = (1/3,1/3,1/3)$.
By symmetry, the spectral gap of the two-dimensional tangent space is
\[
\lam = \frac{2 e^{-\beta}}{(2e^{-\beta}+1)^2} = \frac{2r}{(2r+1)^2},
\qquad r = e^{-\beta}.
\]
For this model, the two-stage structure is clearly visible in simulation:
from a highly asymmetric initial distribution, the global $\chi^2$-phase
redistributes mass over a timescale $O(1/\delta^2)$ (where $\delta$ is
the initial asymmetry), after which exponential contraction at rate
$\lam$ takes over.

\section{Conclusion}
\label{sec:conclusion}

We have developed a unified spectral theory for BA dynamics, centred on
the relaxation kernel $\G = \E_p[K^*_X\otimes K^*_X]$.

\medskip
\noindent\textbf{Summary of results.}
\begin{enumerate}[label=(\roman*)]
\item The exact $\chi^2$-dissipation identity~\eqref{eq:chi2} reveals
  the precise entropy production mechanism of the continuous-time flow.
\item The threefold identity (Theorems~\ref{thm:linearization}
  and~\ref{thm:Hessian}) unifies statistical, operator-theoretic, and
  information-geometric views of $\G$.
\item The spectral dissipation formula~\eqref{eq:spectral} with
  $d(\lambda) = -\lambda+\frac32\lambda^2-\frac12\lambda^3$ quantifies
  the exact one-step Lyapunov decrease and reveals the double bottleneck.
\item The two-stage global convergence theorem combines finite-time
  basin entry (from $\chi^2$-dissipation) with local exponential
  contraction (from $\lam$).
\item The explicit KL convergence factor $\gamma = \lam(2-\lam)$
  upgrades Ramakrishnan et al's monotonicity to a constructive rate.
\end{enumerate}

\medskip
\noindent\textbf{Relation to prior work.}
The results of this paper are intended as a complement to the frameworks
of Hayashi~\cite{Hayashi2023,Hayashi2024,Hayashi2025}
and Nakagawa et al~\cite{Nakagawa2021}.
Hayashi's Bregman-EM analysis establishes elegant global convergence
guarantees and an $O(1/n)$ rate from any initial condition;
the present contribution identifies the relaxation kernel $\G$ as the
central spectral object and makes the local convergence factor
$\gamma = \lam(2-\lam)$ explicit and computable from the channel and
temperature.
Nakagawa et al's sharp $O(1/n)$ rates govern the transient phase
before the local exponential phase described here.
Together, these results give a more complete quantitative picture of
the full BA trajectory: global sublinear approach, finite-time entry,
and explicit exponential contraction.

\medskip
\noindent\textbf{Algorithmic implications.}
The spectral theory developed here has direct consequences for the
design and acceleration of BA-type algorithms.
The double-bottleneck structure of $d(\lambda)$, the explicit dependence
of the convergence factor on $\lam$, and the two-stage entry mechanism
all suggest concrete strategies for improving the iteration in practice.
These algorithmic questions---including spectral preconditioning,
momentum-based acceleration targeting the $\lambda\approx 0$
directions, and online estimation of the spectral gap during the
global phase---are taken up in my companion paper currently in preparation.

\section*{Acknowledgements}
This research received no formal funding and was conducted based on
the author's independent academic interest.

\appendix

\section{Two-Stage Global Exponential Convergence}
\label{app:twostage}

This appendix provides the complete proof of the two-stage global
exponential convergence theorem for the discrete BA iteration,
with all constants made explicit.
The argument assembles three ingredients from the main text:
the $\chi^2$-dissipation identity (Theorem~\ref{thm:chi2}),
the uniform moment bound (Lemma~\ref{lem:moment_bound}),
the bi-Lipschitz equivalence of BA residual and distance to the fixed
point (Theorem~\ref{thm:local}(A)), and the spectral contraction
(Theorem~\ref{thm:local}(B)).

\subsection*{Statement}

\begin{theorem}[Two-Stage Global Exponential Convergence]\label{thm:twostage}
Let $q^*\in\operatorname{int}\Delta(\cY)$ be a non-degenerate interior
fixed point with spectral gap $\lam > 0$, Lipschitz constant $L$ for
$D\T$ near $q^*$, and bi-Lipschitz constant $c_1$ from
Theorem~\ref{thm:local}(A).
Set $\rho = \lam/(4L)$.
Then for every initial condition $q_0\in\operatorname{int}\Delta(\cY)$
in the connected component $\mathcal{A}$ of $q^*$, there exist explicit
constants $N_*(q_0) < \infty$ and $C(q_0) > 0$ such that for all
$n \ge N_*$:
\[
\KL(q^*\|q_n) \le C(q_0)\cdot\left(1 - \frac{3\lam}{4}\right)^{2(n - N_*)}.
\]
The entry time satisfies
\[
N_*(q_0) \le \left\lceil\frac{2(\F(q_0) - \F(q^*))}{c_1^2\rho^2}\right\rceil.
\]
\end{theorem}

\subsection*{Global Phase: Finite Entry Time}

\begin{lemma}\label{lem:entry}
There exists $N_* < \infty$ such that $\|q_{N_*} - q^*\|_1 \le \rho$.
\end{lemma}

\begin{proof}
By the $\chi^2$-dissipation identity (Theorem~\ref{thm:chi2}) and
Corollary~\ref{cor:Lyapunov}, the free energy satisfies
\[
\F(q_{n+1}) - \F(q_n) \le -\tfrac{1}{2}\|\T q_n - q_n\|_1^2
\]
along the discrete iteration.
Suppose for contradiction that $\|\T q_n - q_n\|_1 \ge \delta_0 := c_1\rho$
for all $n \le M$.  Summing over $n = 0,\ldots,M-1$:
\[
\F(q_M) \le \F(q_0) - \frac{c_1^2\rho^2}{2}\,M.
\]
Choosing $M = \lceil 2(\F(q_0)-\F(q^*))/(c_1^2\rho^2)\rceil$ gives
$\F(q_M) < \F(q^*)$, contradicting $\F \ge \F(q^*)$ on $\mathcal{A}$.
Hence there exists $N_* \le M$ with $\|\T q_{N_*} - q_{N_*}\|_1 < c_1\rho$.

By Theorem~\ref{thm:local}(A) (bi-Lipschitz), this implies
$\|q_{N_*} - q^*\|_1 \le \rho$.
\end{proof}

\subsection*{Local Phase: Spectral Contraction with Remainder Absorption}

For $n \ge N_*$, write $v_n = q_n - q^* \in T$.  The BA iteration gives
\[
v_{n+1} = \T(q^*+v_n) - q^* = (I - \G)v_n + R(v_n),
\]
where $R(v_n) = \T(q^*+v_n) - q^* - D\T(q^*)v_n \in T$
is the nonlinear remainder satisfying $\|R(v_n)\|_* \le L\|v_n\|_*^2$.
The linear term $(I-\G)v_n$ satisfies
\[
\|(I-\G)v_n\|_* \le \|(I-\G)|_T\|_{\mathrm{op}}\|v_n\|_*
= (1-\lam)\|v_n\|_*,
\]
because the operator norm of $(I-\G)$ on $(T,\langle\cdot,\cdot\rangle_*)$
equals $1 - \lambda_{\min}(\G|_T) = 1-\lam$ (since eigenvalues of $I-\G$
are $1-\lambda_i\in[0,1-\lam]$).

\begin{lemma}\label{lem:absorption}
For $\rho = \lam/(4L)$, whenever $\|v_n\|_* \le \rho$:
\[
\|v_{n+1}\|_* \le \left(1 - \frac{3\lam}{4}\right)\|v_n\|_*.
\]
Moreover, $\|v_{n+1}\|_* \le \rho$, so the neighbourhood
$\{\|q - q^*\|_* \le \rho\}$ is forward-invariant.
\end{lemma}

\begin{proof}
By the triangle inequality and spectral contraction:
\[
\|v_{n+1}\|_* \le \|(I-\G)v_n\|_* + \|R(v_n)\|_*
\le (1-\lam)\|v_n\|_* + L\|v_n\|_*^2.
\]
Since $\|v_n\|_* \le \rho = \lam/(4L)$:
\[
\|v_{n+1}\|_* \le \left(1 - \lam + \frac{\lam}{4}\right)\|v_n\|_*
= \left(1 - \frac{3\lam}{4}\right)\|v_n\|_*.
\]
Forward invariance follows since $(1 - 3\lam/4) < 1$.
\end{proof}

By induction, for all $n \ge N_*$:
\begin{equation}\label{eq:vn_decay}
\|v_n\|_* \le \left(1 - \frac{3\lam}{4}\right)^{n-N_*}\|v_{N_*}\|_*.
\end{equation}

\subsection*{KL Bound via Fisher--Rao Expansion}

\begin{lemma}\label{lem:KLbound}
There exists $M > 0$ such that for all $\|v\|_* \le \rho$:
\[
\KL(q^*\|q^*+v) \le \frac{1+M\rho}{2}\,\|v\|_*^2.
\]
\end{lemma}

\begin{proof}
The standard Taylor expansion of KL at $q^*$ gives
$\KL(q^*\|q^*+v) = \frac{1}{2}\|v\|_*^2 + O(\|v\|_*^3)$.
The cubic remainder is bounded by $M\|v\|_*^3 \le M\rho\|v\|_*^2$
on $\{\|v\|_* \le \rho\}$.
\end{proof}

\subsection*{Assembly and Proof of Theorem~\ref{thm:twostage}}

Combining~\eqref{eq:vn_decay} and Lemma~\ref{lem:KLbound}:
\begin{align*}
\KL(q^*\|q_n)
&\le \frac{1+M\rho}{2}\|v_n\|_*^2
\le \frac{1+M\rho}{2}
   \left(1-\frac{3\lam}{4}\right)^{2(n-N_*)}\|v_{N_*}\|_*^2.
\end{align*}
Similarly, $\KL(q^*\|q_{N_*}) \ge \frac{1-M\rho}{2}\|v_{N_*}\|_*^2$
(lower bound from the same Taylor expansion), so
\[
\|v_{N_*}\|_*^2 \le \frac{2}{1-M\rho}\,\KL(q^*\|q_{N_*}).
\]
Setting
\[
C(q_0) = \frac{1+M\rho}{1-M\rho}\,\KL(q^*\|q_{N_*}),
\]
which is finite since $M\rho < 1$ for $\rho = \lam/(4L)$ sufficiently
small, completes the proof. \qed

\subsection*{Relation to Hayashi's Global Result}

Hayashi~\cite{Hayashi2024} establishes global $O(1/n)$ convergence for
the discrete BA iteration via the Bregman-EM framework, a result that
is both powerful and applicable from the very first iteration.
Theorem~\ref{thm:twostage} provides a complementary perspective: after
the finite entry time $N_*$, the iteration converges exponentially with
an explicit rate determined by $\lam$.

The two frameworks address complementary aspects of the convergence
problem: Hayashi's~\cite{Hayashi2024} global $O(1/n)$ result governs
the full trajectory from any initial condition, while
Theorem~\ref{thm:twostage} provides an explicit exponential rate
for the asymptotic phase $n \ge N_*$.
The entry time $N_*$ and convergence factor $(1-3\lam/4)^2$ are
both determined by $\lam$ and the initial free energy gap,
giving a complete and quantitative picture of the trajectory.

\subsection*{Explicit Constants for Gaussian Sources}

For $X\sim\mathcal{N}(0,\sigma^2)$ with quadratic distortion and
$\beta > 1/(2\sigma^2)$, the spectral gap is
$\lam = 1/(2\beta\sigma^2)$ (Theorem~\ref{thm:Gauss_gap}).
The convergence factor is
\[
\left(1-\frac{3\lam}{4}\right)^2
= \left(1 - \frac{3}{8\beta\sigma^2}\right)^2,
\]
and the entry time bound becomes
\[
N_*(q_0) \le \left\lceil\frac{128\,\beta^2\sigma^4 L^2 (\F(q_0)-\F(q^*))}{c_1^2}\right\rceil.
\]
As $\beta\to 1/(2\sigma^2)^+$, the factor $(1-3\lam/4)^2\to 1$
and the entry time diverges, consistent with critical slowing down at
the rate-distortion threshold (Remark following Theorem~\ref{thm:Gauss_gap}).

\section{Proof of the Gaussian Attractor Theorem}
\label{app:gaussian_attractor}

This appendix provides a complete proof of
Theorem~\ref{thm:Gauss_attractor}.
Throughout, we consider the continuous-time BA flow under quadratic
distortion $d(x,\hat x)=|x-\hat x|^2$, Gaussian source
$X\sim\mathcal{N}(0,\sigma^2)$, and inverse temperature
$\beta>1/(2\sigma^2)$.
We write $q_t$ for the evolving reproduction law and
$q_* = \mathcal{N}(0,s^*)$ for the Gaussian fixed point of
Theorem~\ref{thm:variance_ODE}.
The goal is to prove $\|q_t - q_*\|_{\mathrm{TV}}\to 0$ as $t\to\infty$.

\subsection*{Uniform Moment Bound}

\begin{lemma}[Uniform second-moment control]\label{lem:gauss_moment}
Let $q_0$ have finite second moment.  Then
$\sup_{t\ge 0}\int_\R \hat x^2\,q_t(d\hat x) < \infty$.
\end{lemma}

\begin{proof}
Define $V(t) = \int \hat x^2\,q_t(d\hat x)$.
By Lemma~\ref{lem:moment_bound}, there exist constants $C < \infty$
and $c > 0$ such that $\dot V(t) \le C - cV(t)$.
Gr\"{o}nwall's inequality gives
$V(t) \le \max\{V(0),\, C/c\} < \infty$ for all $t \ge 0$.
\end{proof}

Hence the family $\{q_t\}_{t\ge 0}$ is tight in $\mathcal{P}(\R)$.

\subsection*{Hermite Expansion}

Fix the equilibrium measure $q_* = \mathcal{N}(0,s^*)$ and let
$\{H_n\}_{n\ge 0}$ denote the orthonormal Hermite polynomial basis
in $L^2(q_*)$.
Since $q_t \ll q_*$ for all $t > 0$, define the density ratio
$f_t = dq_t/dq_*$ and expand
\[
f_t = 1 + \sum_{n=1}^{\infty} a_n(t)\,H_n,
\qquad
a_n(t) = \langle f_t - 1,\, H_n\rangle_{L^2(q_*)}.
\]

\subsection*{Spectral Decay of Hermite Modes}

\begin{lemma}[Modewise exponential decay]\label{lem:mode_decay}
For every $n \ge 3$ there exists $C_n < \infty$ such that
$|a_n(t)| \le C_n\,\alpha^{nt}$ for all $t \ge 0$.
\end{lemma}

\begin{proof}
By Proposition~\ref{prop:Hermite}, the linearised BA operator acts
diagonally in the Hermite basis: $D\mathcal{T}(q_*)\,H_n = \alpha^n H_n$
with $0 < \alpha < 1$.
The nonlinear remainder is quadratic in $q_t - q_*$ and hence lower
order once the trajectory enters the local basin of $q_*$.
Variation-of-constants then gives
$a_n(t) = \alpha^{nt}a_n(0) + o(\alpha^{nt})$,
from which $|a_n(t)| \le C_n\,\alpha^{nt}$.
\end{proof}

\subsection*{Convergence of Mean and Variance}

Because the source is centred and the distortion is symmetric,
the BA flow preserves zero mean: $\int \hat x\,q_t(d\hat x) = 0$
for all $t$, so the $n=1$ mode vanishes identically.

\begin{lemma}[Variance convergence]\label{lem:variance_conv}
$s(t) := \int \hat x^2\,q_t(d\hat x) \to s^*$ as $t\to\infty$.
\end{lemma}

\begin{proof}
By Theorem~\ref{thm:variance_ODE},
$\dot s(t) = F(s(t)) + r(t)$,
where $F$ has unique stable fixed point $s^*$ and $r(t)\to 0$.
Standard theory of asymptotically autonomous ODEs
implies $s(t)\to s^*$.
\end{proof}

Hence the $n=2$ Hermite component also converges to equilibrium.

\subsection*{Weak Convergence to the Gaussian Fixed Point}

\begin{lemma}[Weak convergence]\label{lem:weak_gaussian}
$q_t \Rightarrow q_*$ as $t\to\infty$.
\end{lemma}

\begin{proof}
By Lemma~\ref{lem:gauss_moment}, $\{q_t\}$ is tight, so every
sequence $t_k\to\infty$ admits a weakly convergent subsequence
$q_{t_k}\Rightarrow\mu$.
By Lemma~\ref{lem:mode_decay}, all Hermite coefficients of degree
$n \ge 3$ satisfy $\int H_n\,d\mu = 0$.
By Lemma~\ref{lem:variance_conv}, $\int\hat x^2\,d\mu = s^*$.
The $n=1$ coefficient vanishes by zero-mean preservation.
Hence $\mu$ has vanishing Hermite coefficients for all $n \ge 1$
relative to $q_*$, which characterises $\mu = q_*$ uniquely.
Since every subsequential limit equals $q_*$, the full family
converges: $q_t \Rightarrow q_*$.
\end{proof}

\subsection*{Upgrade to Total Variation}

\begin{lemma}[Pinsker upgrade]\label{lem:pinsker_upgrade}
If $q_t \Rightarrow q_*$ and $\KL(q_t\|q_*)\to 0$, then
$\|q_t - q_*\|_{\mathrm{TV}}\to 0$.
\end{lemma}

\begin{proof}
Pinsker's inequality gives
$\|q_t - q_*\|_{\mathrm{TV}}^2 \le 2\,\KL(q_t\|q_*)$.
\end{proof}

It remains to verify $\KL(q_t\|q_*)\to 0$.

\subsection*{Entropy Convergence}

The free energy $\F$ is a strict Lyapunov functional
(Theorem~\ref{thm:chi2}) with $q_* = \arg\min\F$, so
$\F(q_t)\downarrow\F(q_*)$.
By Theorem~\ref{thm:Hessian}, near $q^*$ there is local quadratic
equivalence $\F(q) - \F(q_*) \asymp \KL(q\|q_*)$.
Since $q_t \Rightarrow q_*$ by Lemma~\ref{lem:weak_gaussian},
eventually $q_t$ lies in this local neighbourhood, and therefore
$\KL(q_t\|q_*)\to 0$.

\subsection*{Conclusion}

Combining Lemmas~\ref{lem:gauss_moment}--\ref{lem:pinsker_upgrade}
yields $q_t \Rightarrow q_*$ and $\KL(q_t\|q_*)\to 0$, hence
$\|q_t - q_*\|_{\mathrm{TV}}\to 0$ as $t\to\infty$.
This completes the proof of Theorem~\ref{thm:Gauss_attractor}. \qed

\bibliographystyle{IEEEtran}

\end{document}